\documentclass[11pt]{scrartcl}

\usepackage{url}
\usepackage[hidelinks]{hyperref}
\usepackage[utf8]{inputenc}
\usepackage[small]{caption}
\usepackage{graphicx}
\usepackage{amsmath}
\usepackage{booktabs}
\usepackage{amsthm}
\usepackage{amssymb}

\usepackage{cleveref}
\usepackage{color}
\usepackage{xspace}
\usepackage{tikz}
\usepackage{tabularx}
\usepackage{pbox}
\usepackage{framed}
\usepackage{enumitem}
\usepackage{arydshln}

\usepackage{algorithm,algpseudocode}

\DeclareOldFontCommand{\bf}{\normalfont\bfseries}{\mathbf}

\usepackage{pifont,xcolor}

\newcommand{\egalratio}{\text{\upshape Egal-ratio}_{n,\alpha}(\mech)\xspace}
\newcommand{\epschange}{$\varepsilon$-change\xspace}
\newcommand{\mech}{\mathcal{M}}
\newcommand{\alloc}{\mathcal{A}}
\newcommand{\allocAdd}{A_{\text{add}}\xspace}
\newcommand{\allocRemove}{A_{\text{remove}}\xspace}
\newcommand*\diff{\mathop{}\!\mathrm{d}} 

\newtheorem{theorem}{Theorem}[section]
\newtheorem{corollary}[theorem]{Corollary}
\newtheorem{proposition}[theorem]{Proposition}
\newtheorem{lemma}[theorem]{Lemma}

\theoremstyle{definition}
\newtheorem{definition}[theorem]{Definition}
\newtheorem{example}[theorem]{Example}

\usepackage{natbib}

\allowdisplaybreaks

\begin{document}

\title{Truthful Cake Sharing}

\author{
Xiaohui Bei\\Nanyang Technological University
\and
Xinhang Lu\\ UNSW Sydney
\and
Warut Suksompong\\National University of Singapore
}

\date{\vspace{-3ex}}

\maketitle

\begin{abstract}
The classic cake cutting problem concerns the fair allocation of a heterogeneous resource among interested agents.
In this paper, we study a public goods variant of the problem, where instead of competing with one another for the cake, the agents all share the same subset of the cake which must be chosen subject to a length constraint.
We focus on the design of truthful and fair mechanisms in the presence of strategic agents who have piecewise uniform utilities over the cake.
On the one hand, we show that the leximin solution is truthful and moreover maximizes an egalitarian welfare measure among all truthful and position oblivious mechanisms.
On the other hand, we demonstrate that the maximum Nash welfare solution is truthful for two agents but not in general.
Our results assume that mechanisms can block each agent from accessing parts that the agent does not claim to desire; we provide an impossibility result when blocking is not allowed.
\end{abstract}

\section{Introduction}
\label{sec:intro}

A fundamental problem in social choice theory is the fair allocation of scarce resources among multiple agents.
When the resource is heterogeneous and divisible, this problem is commonly known as \emph{cake cutting}, with the cake serving as a metaphor for the heterogeneous resource.
Cake cutting has been extensively studied for over half a century in mathematics and economics, and more recently in computer science \citep{BramsTa96,RobertsonWe98,Procaccia16}.

In this paper, we consider a variant of the classic cake cutting problem where instead of competing with one another for the cake, the agents all share the same subset of the cake, which must be chosen subject to a length constraint. We refer to this setting as \emph{cake sharing}.
The cake sharing problem captures many real-world scenarios, such as when a group of agents need to decide the time periods for which they should reserve a sports facility or a conference room for collective use given their limited budget, or when a group of users seek to agree upon the files to store in a shared cache memory.
Our goal is to design cake sharing mechanisms that are both \emph{truthful} and \emph{fair}.
Truthfulness requires that it should be in every agent's best interest to report her true underlying preferences to the mechanism.
A truthful mechanism makes it easy for agents to participate in, as they do not have to act strategically and reason about beneficial manipulations; it also simplifies the job of the mechanism designer when reasoning about the possible behavior of the agents.
Note that truthfulness by itself is easy to obtain, for example by ignoring the agents' reports completely and allocating a prespecified subset of the cake.
However, this is a patently unfair mechanism, as it leaves any agent who has no value for that subset empty-handed.
Is there a mechanism that is truthful and at the same time satisfies a certain degree of fairness for all agents?

Two mechanisms that have been used in various resource allocation settings and often shown to exhibit attractive fairness properties are the \emph{maximum Nash welfare (MNW) solution} and the \emph{leximin solution}.
The MNW solution chooses an allocation that maximizes the product of the agents' utilities among all feasible allocations.
The leximin solution considers all feasible allocations that maximize the minimum among the agents' utilities; among all such allocations, it considers those maximizing the second smallest utility, and so on.
Due to their optimization nature, both solutions fulfill an important economic efficiency criterion of \emph{Pareto optimality}: there is no other feasible outcome that makes some agent better off and no agent worse off compared to the chosen outcome.
Indeed, any such improved outcome would also be an improvement with respect to the corresponding optimization objective.
Given the broad appeal of the two mechanisms, are they appropriate choices for our cake sharing setting, especially from the truthfulness perspective?

\subsection{Our Results}
\label{sec:our-results}

As is standard in the cake cutting literature, we model the cake as an interval $[0,1]$; for a given parameter $\alpha\in[0,1]$, a subset of length at most $\alpha$ of the cake can be collectively allocated to the agents.
We assume that the agents have \emph{piecewise uniform} utilities, meaning that each agent has a desired subset of the cake which she values uniformly.
Except in Section~\ref{sec:no-blocking}, we also assume that once a mechanism chooses a subset of the cake, it can ``block'' each agent from accessing certain parts of the cake, usually those that the agent does not desire according to her report.
We remark here that blocking can be easily implemented in the aforementioned applications of cake sharing, for example by disallowing agents from accessing part of the cache memory that they do not demand or restricting their access to the sports facility during the times that
they claim to be unavailable.

In Section~\ref{sec:leximin}, we focus on the leximin solution.
Our main technical result establishes the truthfulness of the solution for any number of agents with arbitrary piecewise uniform utilities.
At a high level, our proof proceeds by showing that the leximin solution is immune to certain types of manipulations, and then arguing that this immunity is sufficient to protect the solution against all possible manipulations.
Along the way, we introduce the notion of an \emph{$\varepsilon$-change}---a tiny change from one utility vector or allocation towards another---which may be useful in related settings.
Additionally, we show that each agent receives the same utility in all leximin allocations (which means that tie-breaking is inconsequential) and that such an allocation can be computed in polynomial time.

Since truthfulness by itself can be trivially obtained as we explained earlier, we consider in Section~\ref{sec:egalitarian} the fairness of mechanisms.
We adapt the well-known notion of \emph{proportionality} from cake cutting and measure fairness using the \emph{egalitarian ratio}, which is defined as the worst-case egalitarian welfare that the mechanism provides over all instances, where we normalize the agents' utilities for the entire cake when computing their egalitarian welfare.
We show that for any $\alpha$ and number of agents $n$, the leximin solution has egalitarian ratio exactly $\alpha/(n-(n-1)\alpha)$.
Moreover, we prove that this ratio is already optimal among all mechanisms that are truthful and \emph{position oblivious} (see Definition~\ref{def:position-oblivious}).
Our results in Sections~\ref{sec:leximin} and \ref{sec:egalitarian} establish the leximin solution as an attractive mechanism in the setting of cake sharing.

In Section~\ref{sec:MNW}, we turn our attention to the MNW solution.
We show that the solution is equivalent to the leximin solution in the case of two agents, and is therefore truthful in that case.
In general, however, a result of \citet[Theorem~3]{AzizBoMo20} implies the non-truthfulness of the MNW solution in our setting.
We strengthen their result by showing that MNW is not truthful even when an agent is only allowed to report a subset of her true desired piece.\footnote{We remark that subset manipulation is a highly restricted form of manipulation.
Indeed, \citet{Peters18} noted that reporting a subset is a ``particularly simple fashion'' of manipulating, and used subset manipulation as the ``official notion'' of truthfulness.}
Moreover, in contrast to Aziz et al.'s example, the symmetry structure in our example allows us to provide a relatively short proof of the non-truthfulness that can be easily verified by hand.

Finally, we demonstrate in Section~\ref{sec:no-blocking} that the ability to block is crucial for the truthfulness of mechanisms.
In particular, we show that no truthful, Pareto optimal, and position oblivious mechanism can achieve a positive egalitarian ratio when blocking is not allowed.

\subsection{Related Work}
\label{sec:related}

While the model of cake sharing is new to the best of our knowledge, the selection of a collective subset from a given set subject to a size or budget constraint has been studied in several lines of work.
In \emph{multiwinner voting}, the goal is to choose a certain number of candidates to form a committee, where criteria can include excellence and diversity---see the surveys by \citet{FaliszewskiSkSl17} and \citet{LacknerSk21}.
In that setting, \citet{Peters18} proved that no rule can simultaneously satisfy a form of fairness and a form of truthfulness when agents have approval preferences (analogous to piecewise uniform utilities in our setting).
A key difference between multiwinner voting and cake sharing is that the candidates in the former are discrete and cannot be divided into arbitrarily small pieces.
Variants where discrete items instead of candidates are selected have also been considered \citep{SkowronFaLa16,ManurangsiSu19}.

A long list of recent papers have addressed the problem of \emph{participatory budgeting}, where the citizens decide how a public budget should be spent on possible projects in their community---see the survey by \citet{AzizSh21}.
Some models assume that projects are discrete (each project can either be fully completed or not at all), while others assume that they are divisible (partial completion of a project yields some utility to the citizens).
In either case, there is a prespecified set of projects and the preference of an agent within a project is uniform, so participatory budgeting cannot capture our cake sharing model where there is no predetermined division of the cake into homogeneous units.

\citet{AzizBoMo20} studied a probabilistic voting setting in which agents have dichotomous preferences over $m$ alternatives and the goal is to output a probability distribution over the alternatives.
Their model corresponds to a special case of our model where for each $j=1,\dots,m$, the interval $[(j-1)/m, j/m]$ represents alternative~$j$, and $\alpha = 1/m$; in this special case, agents are not allowed to have ``breakpoints'' that are not multiples of $1/m$ in their utility functions (see the precise definition of a breakpoint in Section~\ref{sec:prelim}).
Like us, Aziz et al.~showed that the leximin solution is truthful.\footnote{They called the notion \emph{excludable strategyproofness}, which is equivalent to truthfulness with blocking in our setting.}
Our results on the leximin solution generalize and strengthen theirs in three important ways.
First, we allow agents to report \emph{arbitrary} breakpoints---this considerably enlarges the strategy space of the agents and introduces an aspect that cannot be captured by their model.
Second, our model allows an arbitrary value of $\alpha$ instead of only $\alpha = 1/m$.
Third, we establish a tight bound on the egalitarian ratio and show that the leximin solution achieves this bound.
Therefore, we believe that overall, our results make a significantly stronger case in favor of the leximin solution.

\citet{FriedmanGkPs19} investigated a model in which agents share a cache memory unit, focusing on truthfulness and fairness like we do.
In their model, each agent has a private file that no other agent is interested in, and there is a large public file that may be of interest to multiple agents.
The challenge of the mechanism is to elicit the true ratio between each agent's utility for the public file and that for her private file.
These authors demonstrated that the ability to block can also help mechanisms achieve better guarantees in their setting, in particular by preventing ``free riding''.

Truthfulness in cake \emph{cutting} has been considered in several papers \citep{MayaNi12,ChenLaPa13,KurokawaLaPr13,AzizYe14,BranzeiMi15,BeiChHu17,MenonLa17,BeiHuSu20,Tao21}.
Like our paper, a number of these papers also address the case of piecewise uniform utilities.
Two important fairness properties in cake cutting are \emph{envy-freeness} and \emph{proportionality}.
Note that envy-freeness is always fulfilled in our setting (as long as the mechanism does not block any agent's valued cake), since all agents share the same subset of the cake.
On the other hand, proportionality has a similar flavor as our egalitarian ratio notion, where we want to guarantee a certain level of utility for every agent.

Finally, both the leximin and MNW solutions have been examined in a variety of settings and often shown to exhibit desirable properties \citep{BogomolnaiaMo04,KurokawaPrSh18,CaragiannisKuMo19,SegalhaleviSz19,AzizBoMo20,HalpernPrPs20,PlautRo20,BrandlBrPe20}.

\section{Preliminaries}
\label{sec:prelim}

Our setting includes a set of agents denoted by $N = \{1, 2, \dots, n\}$ and a heterogeneous divisible good (or \emph{cake}) represented by the normalized interval $[0, 1]$.
A \emph{piece of cake} is a union of finitely many disjoint (closed) intervals.
Denote by $\ell(I)$ the length of an interval $I$, that is, $\ell([a, b]) = b - a$.
For a piece of cake $S$ consisting of a set of intervals $\mathcal{I}_S$, we denote $\ell(S) = \sum_{I \in \mathcal{I}_S} \ell(I)$.
Each agent $i \in N$ is endowed with a \emph{density function} $f_i \colon [0, 1] \to \mathbb{R}_{\geq 0}$, which captures how the agent values different parts of the cake.
We assume that the agents have \emph{piecewise uniform} utilities: for each agent~$i$, each part of the cake is either desired or undesired, and the density function $f_i$ takes on the value $1$ for all desired parts and $0$ for all undesired parts.
Let $W_i \subseteq [0, 1]$ denote the (not necessarily contiguous) piece of cake on which $f_i = 1$.
The utility of agent $i$ for any piece of cake $S$ is given by $u_i(S) := \ell(S\cap W_i)$.\footnote{Some cake cutting papers normalize the utility functions so that $u_i([0,1])=1$ for all $i\in N$; we do not follow this convention.
See also the end of Section~\ref{sec:leximin} for a discussion of normalization.}
We assume that $\ell(W_i) > 0$ for every~$i$, since we can simply ignore an agent~$i$ with $\ell(W_i) = 0$.

Let $\alpha\in[0,1]$ be a given parameter.
We refer to a setting with agents, their density functions, and the parameter $\alpha$ as an \emph{instance}.
A \emph{mechanism}~$\mech(R)$ chooses from any given instance~$R$ a piece of cake~$A$ with $\ell(A)\le\alpha$.
However, this does not mean all agents have full access to $A$, because we allow the mechanism to \emph{block} each agent from accessing certain parts of the selected piece.
Specifically, after choosing $A$, the mechanism assigns piece $A_i\subseteq A$ to agent~$i$; we call $\alloc=(A,A_1,\dots,A_n)$ an \emph{allocation}.
The utility of agent~$i$ from the allocation~$\alloc$ is $u_i(A_i)$.
Since the cases $\alpha=0$ and $\alpha=1$ are trivial,
 we assume from now on that $\alpha\in(0,1)$.
Given an instance, every point that is a left or right endpoint of an interval in $W_i$ for at least one $i$ is called a \emph{breakpoint}; the points $0$ and $1$ are also considered to be breakpoints.
Observe that for any instance, the agents' utilities for a piece of cake $S$ depend only on the amounts of cake between consecutive pairs of breakpoints included in $S$.

We now define the central property of our paper.

\begin{definition}[Truthfulness]
\label{def:truthful}
A mechanism is \emph{truthful} if for any instance $R$ with $\mech(R) = (A, A_1, \dots, A_n)$ and any agent $i\in N$, if the agent reports $W_i'\ne W_i$ and the mechanism returns the allocation $\alloc'=(A',A_1',\dots,A_n')$ on the modified instance, then $u_i(A_i)\ge u_i(A_i')$.
\end{definition}

Next, we define the two main mechanisms in this paper.

\begin{definition}[Leximin]
\label{def:leximin}
Given an instance, the \emph{leximin solution} considers pieces of cake $A$ with $\ell(A)\le\alpha$ such that the minimum among the utilities $u_1(A),\dots,u_n(A)$ is maximized; among all such pieces $A$, it considers those for which the second smallest utility is maximized, and so on, until after considering the largest utility, it chooses one of the pieces $A$ that remain.
It then assigns $A_i = A\cap W_i$ for all $i\in N$.
\end{definition}

\begin{definition}[MNW]
\label{def:MNW}
Given an instance, the \emph{maximum Nash welfare (MNW) solution} chooses a piece of cake $A$ with $\ell(A)\le\alpha$ such that the product $\prod_{i\in N}u_i(A)$ is maximized.
It then assigns $A_i = A\cap W_i$ for all $i\in N$.
\end{definition}

The following example illustrates some of our definitions.

\begin{example}\label{ex:illustration}
Let $\alpha = 1/2$.
Consider an instance with two agents whose utility functions are given as follows:
\[
W_1 = [0, 1/2], \qquad W_2 = [1/4, 7/8].
\]
A possible piece $A$ selected by both leximin and MNW is $A = [1/8, 5/8]$.\footnote{We show in Theorem~\ref{thm:MNW-leximin-equivalence} that leximin and MNW are equivalent in the case of two agents.}
Then, agent~$1$ has access to the piece $A_1 = A \cap W_1 = [1/8, 1/2]$ while agent~$2$ has access to the piece $A_2 = A \cap W_2 = [1/4, 5/8]$.
Both agents receive utility $3/8$.
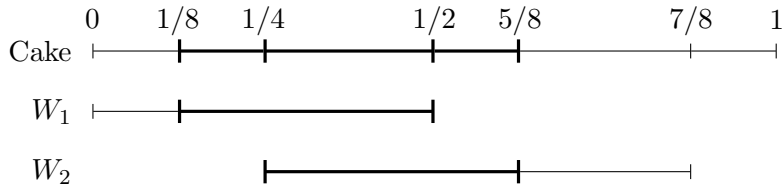
\begin{center}
\begin{tikzpicture}[x=9cm,y=0.8cm]
\node[label=left:{Cake}] at (0,0) {};
\draw[|-] (0,0) to (0.125,0);
\node at (0,0.531) {$0$};
\node at (0.125,0.5) {$1/8$};
\draw[|-,very thick] (0.125,0) to (0.25,0);
\node at (0.25,0.5) {$1/4$};
\draw[|-|,very thick] (0.25,0) to (0.5,0);
\node at (0.5,0.5) {$1/2$};
\draw[-|,very thick] (0.5,0) to (0.625,0);
\node at (0.625,0.5) {$5/8$};
\draw[-|] (0.625,0) to (0.875,0);
\node at (0.875,0.5) {$7/8$};
\draw[-|] (0.875,0) to (1,0);
\node at (1,0.527) {$1$};

\node[label=left:{$W_1$}] at (0,-1) {};
\draw[|-] (0,-1) to (0.125,-1);
\draw[|-|,very thick] (0.125,-1) to (0.5,-1);

\node[label=left:{$W_2$}] at (0,-2) {};
\draw[|-|,very thick] (0.25,-2) to (0.625,-2);
\draw[-|] (0.625,-2) to (0.875,-2);
\end{tikzpicture}
\end{center}
\end{example}

Since both leximin and MNW always choose $A_i = A\cap W_i$ for all $i\in N$, we can represent an allocation $\alloc$ simply by the set~$A$ when we discuss these mechanisms.
Note that $u_i(A_i) = \ell(A_i\cap W_i) = \ell(A\cap W_i) = u_i(A)$, so it also suffices to consider the agents' utilities with respect to $A$.
By a standard compactness argument and our observation above that the agents' utilities depend only on the amounts of cake between breakpoints, both solutions are well-defined (i.e., the desired maxima are attained).
There may be several maximizing allocations $A$ to choose from, in which case we generally allow arbitrary tie-breaking---as we will see later, this tie-breaking does not influence the utility that each agent receives and therefore does not play a significant role.
We call an allocation that is returned by the MNW solution (resp., leximin solution) under some tie-breaking an \emph{MNW allocation} (resp., leximin allocation).
By our assumptions that $\alpha > 0$ and $\ell(W_i) > 0$ for every~$i$, all MNW allocations and leximin allocations give every agent a strictly positive utility.

\section{Leximin Solution}
\label{sec:leximin}

In this section, we consider the leximin solution.
We begin by establishing basic properties of the solution.

Our first result is that the utility of each agent is the same in all leximin allocations, which means that tie-breaking is not an important issue.
The proof proceeds by assuming for contradiction that two leximin allocations give some agent different utilities, and arguing that the ``average'' of these two allocations would have been a better choice with respect to the leximin ordering.

\begin{proposition}
\label{prop:leximin-unique}
Given any instance, for each agent~$i$, the utility that $i$ receives is the same in all leximin allocations.
\end{proposition}

\begin{proof}
Assume for contradiction that two leximin allocations, $A$ and $A'$, give some agent different utilities.
Let $A''$ be an allocation such that for each pair of consecutive breakpoints, the amount of cake between those breakpoints included in $A''$ is the average of the corresponding amounts for $A$ and $A'$.
By linearity, $A''$ is a feasible allocation, and $u_i(A'') = \frac{1}{2}(u_i(A) + u_i(A'))$ for every $i\in N$.

Since the leximin ordering is a total order, the multiset of utilities that the $n$ agents receive in $A$ must be the same as the corresponding multiset in $A'$.
Let $j$ be an agent with the smallest $\min\{u_j(A), u_j(A')\}$ such that $u_j(A)\ne u_j(A')$, and assume without loss of generality that $u_j(A) < u_j(A')$.
All agents $k$ with $\min\{u_k(A), u_k(A')\} < \min\{u_j(A), u_j(A')\}$ have $u_k(A) = u_k(A')$, and this latter quantity is also equal to $u_k(A'')$.
On the other hand, we have $u_j(A'') = \frac{1}{2}(u_j(A) + u_j(A')) > u_j(A)$, which means that the number of agents who receive utility exactly $u_j(A)$ in $A''$ is strictly less than the corresponding numbers for $A$ and $A'$.
Hence, $A''$ is a better allocation with respect to the leximin ordering than $A$ and $A'$, a contradiction.
\end{proof}

Next, we show that a leximin allocation can be computed efficiently via a linear programming-based approach similar to the one used by \citet{AiriauAzCa19} in the context of portioning.
Recall that in our setting, the utility functions of the agents can be described explicitly by the sets $W_i$.

\begin{proposition}
\label{prop:leximin-algorithm}
There exists an algorithm that computes a leximin allocation in time polynomial in the input size.
\end{proposition}

\begin{proof}
First, we divide the cake into a set $M$ of intervals $I_1,\dots,I_m$ using all breakpoints, so each interval is either desired in its entirety or not desired at all by each agent.
Let $x_j$ denote the length of the interval $I_j$ that we include in our allocation.
Thus, the utility of agent~$i$ can be written as $\sum_{j=1}^m \mathbf{1}_{I_j\subseteq W_i}x_j$, where $\mathbf{1}_X$ denotes the indicator variable for event~$X$.

We proceed by formulating linear programs.
Initially, the set $N'$ of agents whose utility we have already fixed is empty.
We determine the smallest utility in a leximin allocation by solving for the maximum $t^*$ such that the utility of every agent is at least $t^*$.
We then determine an agent who receives utility $t^*$ in a leximin allocation---to this end, for each agent, we solve for the maximum $\varepsilon$ such that this agent receives utility at least $t^*+\varepsilon$ and every other agent receives utility at least $t^*$.
We choose an agent who returns $\varepsilon = 0$, fix the utility of this agent~$i'$ by setting $t_{i'} = t^*$, and continue by finding the next smallest utility among the remaining agents.
The pseudocode of our algorithm is given as Algorithm~\ref{alg:leximin}.

\begin{algorithm}[!ht]
\caption{Computing a leximin allocation}\label{alg:leximin}
\begin{algorithmic}[1]
\Procedure{Leximin$(N, M, \{W_i\}_{i\in[n]})$}{}
\State $N'\leftarrow \emptyset$
\State $t_i\leftarrow 0, \forall i\in N$
\While{$N'\ne N$}
\State Solve the following linear program:

maximize $t^*$ subject to

$\sum_{j=1}^m \mathbf{1}_{I_j\subseteq W_i}\cdot x_j\ge t^*$ \quad $\forall i\in N\setminus N'$

$\sum_{j=1}^m \mathbf{1}_{I_j\subseteq W_i}\cdot x_j= t_i$ \quad\kern 0.5mm  $\forall i\in N'$

$\sum_{j=1}^m x_j\le \alpha$

$0\le x_j\le \ell(I_j)$ \qquad\qquad\hspace{0.5mm} $\forall j\in\{1,\dots,m\}$
\State Set $t^*$ to be the solution of the linear program.
\For{$i'\in N\setminus N'$}
\State Solve the following linear program:

maximize $\varepsilon$ subject to

$\sum_{j=1}^m \mathbf{1}_{I_j\subseteq W_{i'}}\cdot x_j\ge t^* + \varepsilon$

$\sum_{j=1}^m \mathbf{1}_{I_j\subseteq W_i}\cdot x_j\ge t^*$ \quad $\forall i\in N\setminus N'$

$\sum_{j=1}^m \mathbf{1}_{I_j\subseteq W_i}\cdot x_j= t_i$ \quad\kern 0.5mm  $\forall i\in N'$

$\sum_{j=1}^m x_j\le \alpha$

$0\le x_j\le \ell(I_j)$ \qquad\qquad\hspace{0.5mm} $\forall j\in\{1,\dots,m\}$
\If{$\varepsilon = 0$}
\State $N'\leftarrow N'\cup\{i'\}$
\State $t_{i'}\leftarrow t^*$
\EndIf
\EndFor
\EndWhile
\State \Return the solution $\mathbf{x}$ of the last linear program solved
\EndProcedure
\end{algorithmic}
\end{algorithm}

Since our algorithm requires solving $O(n^2)$ linear programs, it runs in polynomial time.
We now establish its correctness.
Consider the first iteration of the \textbf{while}-loop, and the returned value $t^*$ of the first linear program.
We claim that for at least one $i'\in N$, the linear program for $i'$ returns $\varepsilon = 0$.
Indeed, if this is not the case, then for every $i'$, there is a feasible allocation that gives $i'$ a utility strictly greater than $t^*$ and gives every other agent a utility of at least $t^*$; by taking the ``average'' of all such allocations similarly to the proof of Proposition~\ref{prop:leximin-unique}, we obtain a feasible allocation that gives every agent strictly greater than $t^*$, contradicting the definition of $t^*$.
For $i'$ such that $\varepsilon = 0$, we therefore have that the utility of $i'$ is equal to $t^*$ in every leximin allocation.
We then apply a similar argument for the remaining $n-1$ iterations to conclude that the utility of each agent in the allocation that the algorithm returns is equal to
the corresponding utility in every leximin allocation.
It therefore follows that the returned allocation is a leximin allocation.
\end{proof}

We now come to our main result of this section, which establishes the truthfulness of the leximin solution.

\begin{theorem}
\label{thm:leximin-truthful}
For arbitrary tie-breaking, the leximin solution is truthful.
\end{theorem}

At a high level, the proof of Theorem~\ref{thm:leximin-truthful} proceeds by identifying specific types of manipulations, arguing that such manipulations cannot be beneficial when the leximin solution is used, and then showing that being immune to these manipulations implies being immune to all manipulations.
We start by defining an \epschange, a useful concept in our proof.

\begin{definition}[\epschange]
Given two vectors of real numbers $\mathbf{x} = (x_1, x_2, \dots, x_n)$ and $\mathbf{x}' = (x_1', x_2', \dots, x_n')$, an \epschange from $\mathbf{x}$ towards $\mathbf{x}'$ refers to the following continuous operation: for each $i \in \{1, 2, \dots, n\}$, $x_i$ changes linearly to $x_i'' := x_i + \varepsilon(x_i'-x_i)$, where  $\varepsilon$ is sufficiently small so that if $x_i < x_j$, then $x_i'' < x_j''$.
\end{definition}

For ease of expression, we will also use an \epschange to refer to the outcome of such an operation, i.e., the vector $\mathbf{x}''$.
When we discuss $\varepsilon$-changes, we will not specify the exact value of $\varepsilon$: any $\varepsilon$ satisfying the above condition works.
The following lemma establishes a useful property of $\varepsilon$-changes.

\begin{lemma}\label{lem:eps-change-property}
Given two vectors $\mathbf{x}$ and $\mathbf{y}$, if $\mathbf{y}$ is a better vector with respect to the leximin ordering than $\mathbf{x}$, then an \epschange from $\mathbf{x}$ to $\mathbf{y}$ is also a leximin improvement.
\end{lemma}

\begin{proof}
Sort the numbers of $\mathbf{x}$ in non-descending order and group them into buckets so that numbers within each bucket are the same and those in different buckets are different.
Observe that an \epschange improves $\mathbf{x}$ with respect to the leximin ordering if and only if for the lowest bucket where there is a change, some number increases and no number decreases.

Consider an \epschange from $\mathbf{x}$ towards a better leximin vector $\mathbf{y}$.
If some number in the lowest bucket of $\mathbf{x}$ decreases, then $\mathbf{y}$ would not be a leximin improvement of $\mathbf{x}$, so no number in this bucket decreases.
If some number in this bucket increases, we are done by the above observation.
Else, there is no change in this bucket, and we move on to the next bucket and repeat the same argument.
Because $\mathbf{x}$ and $\mathbf{y}$ are different, there must be a change in at least one bucket, which gives our desired conclusion.
\end{proof}

We now extend the definition of an \epschange to allocations.
Recall that for the leximin solution, it suffices to consider the set $A$ instead of the entire allocation $\alloc$.
Given two allocations $A$ and $A'$, an \epschange from $A$ towards $A'$ can be captured by dividing the cake into intervals according to the breakpoints and changing $A$ towards $A'$ so that the length of cake included in the allocation in each interval changes linearly.
Note that when we perform an \epschange from $A$ towards $A'$, by linearity, we also obtain a corresponding \epschange from the vector $(u_1(A),\dots,u_n(A))$ towards $(u_1(A'),\dots,u_n(A'))$, and any allocation obtained during the process is feasible.

Next, we present auxiliary lemmas used for proving the truthfulness of leximin solution.
These lemmas discuss how the leximin allocation can change when an agent modifies her density function in various ways.
For notational convenience, in these lemmas we assume that instance $R$ (resp., $R'$) contains the density functions corresponding to $W_1,\dots,W_n$ (resp., $W_1',\dots,W_n'$).
Our first lemma says that whenever an agent shrinks her desired piece in such a way that it contains the entire portion she receives, then she should still receive the same portion in the new instance.

\begin{lemma}\label{lem:discard-unselected}
Given a leximin allocation $A$ for instance $R$, let $R'$ be an instance such that $A\cap W_i \subseteq W_i' \subseteq W_i$ for an agent $i \in N$ and $W_j' = W_j$ for all $j \in N \setminus \{i\}$.
Then, $A$ is also a leximin allocation for $R'$.
\end{lemma}

\begin{proof}
Suppose for contradiction that $A$ is not a leximin allocation for $R'$.
Consider a leximin allocation $A'$ for $R'$.
Since $W_i'\subseteq W_i$, the utility of agent~$i$ for $A'$ in $R'$ is at most that for $A'$ in $R$.
Moreover, the utility of every agent $j\ne i$ for $A'$ is the same in $R'$ as in $R$.
On the other hand, since $A\cap W_i \subseteq W_i'$, the utility of every agent for $A$ is the same in $R'$ as in $R$.
By our assumption that $A'$ is a better allocation with respect to the leximin ordering than $A$ in $R'$, the same must also hold for $R$, a contradiction.
\end{proof}

Our second lemma says that when an agent shrinks her desired piece, she should not get a higher utility than before.

\begin{lemma}\label{lem:shrink-get-no-more}
Given a leximin allocation $A$ for instance $R$, let $R'$ be an instance such that $W_i' \subseteq W_i$ for an agent $i \in N$ and $W_j' = W_j$ for all $j \in N \setminus \{i\}$.
Let $A'$ be a leximin allocation for $R'$.
Then, $\ell(A'\cap W_i') \leq \ell(A\cap W_i)$.
\end{lemma}

\begin{proof}
We now proceed to prove Lemma~\ref{lem:shrink-get-no-more}.
Suppose for contradiction that $\ell(A'\cap W_i') > \ell(A\cap W_i)$; let $x'$ and $x$ denote the former and latter quantites, respectively.
By Proposition~\ref{prop:leximin-unique}, agent~$i$ receives the same utility in all leximin allocations, so $A$ is a better allocation with respect to the leximin ordering than $A'$ in $R$.
When changing from $A'$ to $A$, since $W_i'\subseteq W_i$, the utility of agent~$i$ decreases from at least $x'$ to $x$ with respect to $R$.
Hence, even if the utility of agent~$i$ started at exactly $x'$ and decreased to $x$, the change would still be a leximin improvement.

Now, consider the agents' utilities with respect to $R'$.
When changing from $A'$ to $A$, since $W_i'\subseteq W_i$, the utility of agent~$i$ decreases from  $x'$ to at most $x$.
Since all changes are in the same direction as the previous change starting from $x'$, which is a leximin improvement, by Lemma~\ref{lem:eps-change-property}, an \epschange from $A'$ towards $A$ is also a leximin improvement with respect to $R'$.
This contradicts the assumption that $A'$ is a leximin allocation for $R'$.
\end{proof}

Our third lemma says that if an agent is already getting her entire desired piece, then whenever she shrinks her desired piece, she should still be at maximum utility.

\begin{lemma}\label{lem:further-discard}
Given a leximin allocation $A$ for instance $R$ with $W_i\subseteq A$ for an agent $i\in N$, let $R'$ be an instance such that $W_i' \subseteq W_i$ and $W_j' = W_j$ for all $j \in N \setminus \{i\}$.
Let $A'$ be a leximin allocation for $R'$.
Then, $W_i'\subseteq A'$.
\end{lemma}

\begin{proof}
Suppose for contradiction that $A'$ does not contain the entire $W_i'$, and let $x' = \ell(A'\cap W_i') < \ell(W_i')$.
By Proposition~\ref{prop:leximin-unique}, agent~$i$ receives the same utility in all leximin allocations, so $A$ is a better allocation with respect to the leximin ordering than $A'$ in $R$.
By Lemma~\ref{lem:eps-change-property}, an \epschange from $A'$ towards $A$ is also a leximin improvement with respect to $R$, so by the characterization of \epschange improvements in the proof of the lemma, in the lowest bucket where there is a change, some number increases and no number decreases.
In this \epschange, the utility of agent~$i$ increases from at least $x'$ towards $\ell(W_i)$.

Now, consider the same \epschange from $A'$ towards $A$, but with respect to $R'$.
The utility of agent~$i$ increases from exactly $x'$ towards $\ell(W_i')$, while those of other agents change in the same way as before.
Since agent~$i$'s utility starts no higher than before and still increases, one can check that in the lowest bucket where there is a change, again some number increases and no number decreases.
Hence, the characterization of \epschange improvements implies that the change is also a leximin improvement with respect to $R'$.
This contradicts the assumption that $A'$ is a leximin allocation for $R'$.
\end{proof}

We are now ready to prove Theorem~\ref{thm:leximin-truthful}.

\begin{proof}[Proof of Theorem~\ref{thm:leximin-truthful}]
Suppose for contradiction that the leximin solution is not truthful.
This means that there exists an instance $R$ with leximin allocation $A$ such that if agent~$i$ reports $\widehat{W}_i$ instead of $W_i$, a leximin allocation $\widehat{A}$ in the new instance $\widehat{R}$ satisfies $\ell(\widehat{A}\cap \widehat{W}_i\cap W_i) > \ell(A\cap W_i)$.
We will keep the desired pieces $W_j$ of agents $j\in N\setminus\{i\}$ unchanged throughout this proof.

First, consider an instance $\widehat{R}'$ where $\widehat{W}_i' = \widehat{W}_i\cap \widehat{A}$.
By Lemma~\ref{lem:discard-unselected} applied to $\widehat{R}$ and $\widehat{R}'$, $\widehat{A}$ is also a leximin allocation for $\widehat{R}'$.
Next, consider an instance $\widehat{R}''$ in which $\widehat{W}_i'' = \widehat{W}_i\cap \widehat{A}\cap W_i$.
Since $\widehat{W}_i'\subseteq \widehat{A}$, by Lemma~\ref{lem:further-discard} applied to $\widehat{R}'$ and $\widehat{R}''$, any leximin allocation for $\widehat{R}''$ must contain the entire $\widehat{W}_i''$.
Recall that $\ell(\widehat{W}_i'') > \ell(A\cap W_i)$.

Finally, consider the instances $R$ and $\widehat{R}''$.
From the former to the latter, agent~$i$'s desired piece shrinks from $W_i$ to $\widehat{W}_i''\subseteq W_i$.
By Lemma~\ref{lem:shrink-get-no-more}, the agent should not get a higher utility through this shrinking.
However, the agent's utility is $\ell(A\cap W_i)$ before the shrinking, and $\ell(\widehat{W}_i'')$ afterwards.
This is a contradiction.
\end{proof}

Observe that unlike MNW, the leximin solution depends on the normalization of the agents' utilities.
Besides our normalization, another common choice in cake cutting is to normalize the utility of every agent for the whole cake to $1$---this leads to an alternative definition of the leximin solution.
We remark here that this variant of leximin is \emph{not} truthful.
To see this, consider two agents with $W_1 = [0, 1/3]$ and $W_2 = [1/3, 2/3]$, and let $\alpha = 1/3$.
In this instance, the (alternative) leximin solution gives each agent length $1/6$ of the cake.
However, if agent~$2$ misreports that $W_2 = [1/3, 1]$, then it is possible that the agent receives the interval $[1/3, 5/9]$ and therefore length $2/9 > 1/6$ of her valued cake.
Let us emphasize that once we fix a mechanism, whether the mechanism is truthful or not is independent of the normalization, because truthfulness does not depend on the normalization.
Hence, our normalization for leximin is the ‘correct’ one to use if we desire truthfulness.
In the next section, we provide further evidence that this normalization is appropriate by showing that our version of the leximin solution achieves a strong fairness guarantee in terms of egalitarian welfare.

\section{Egalitarian Ratio}
\label{sec:egalitarian}

As we mentioned in the introduction, truthfulness by itself is easy to achieve, for example by always allocating a fixed piece of cake of length $\alpha$.
However, this may leave certain agents with zero utility, a patently unfair outcome.
To measure fairness, we adapt the standard notion of proportionality from cake cutting\footnote{Recall that in cake cutting, proportionality stipulates that every agent receives at least $1/n$ of her value for the entire cake.} and consider the \emph{minimum} among the utilities of all agents.
In order to perform meaningful interpersonal comparisons of utilities, we compare the utilities that the agents receive to their utilities for the entire cake in the following definition.

\begin{definition}[Egalitarian ratio]\label{def:egalitarian-ratio}
Given an instance $R$ and an allocation $\mathcal{A}$, the \emph{egalitarian ratio} of $\mathcal{A}$ is defined as
\[
\text{Egal-ratio-alloc}_R(\mathcal{A}) = \min_{i\in N}\frac{u_i(A_i)}{u_i([0,1])}.
\]
For a mechanism $\mech$ and parameters $n$ and $\alpha$, the \emph{egalitarian ratio} of $\mech$ with respect to $n$ and $\alpha$ is defined as
\[
\text{Egal-ratio}_{n,\alpha}(\mech) = \inf_{R} \text{Egal-ratio-alloc}_R(\mech(R)),
\]
where the infimum is taken over all instances with $n$ agents and parameter~$\alpha$.
\end{definition}

In other words, the egalitarian ratio of $\mech$ with respect to $n$ and $\alpha$ is the smallest ratio between an agent's utility for her piece allocated by $\mech$ and her utility for the entire cake, taken over all instances with parameters $n$ and $\alpha$.
For example, if a mechanism always allocates a fixed piece of length $\alpha$ regardless of the agents' utility functions, then its egalitarian ratio with respect to any $n$ and $\alpha\in(0,1)$ is $0$.
We first present a tight upper bound on the egalitarian ratio.

\begin{proposition}
\label{prop:egal-ratio-upper}
For all $n\ge 1$ and $\alpha\in(0,1)$,
\[
0\le \emph{Egal-ratio}_{n,\alpha}(\mech) \le \alpha
\]
for any mechanism $\mech$.
Moreover, for each inequality, there exists a mechanism $\mech$ such that the inequality is tight.
\end{proposition}

\begin{proof}
The lower bound of $0$ holds trivially, and is achieved by the mechanism discussed before the proposition.

For the upper bound, note that if some agent $i$ values the whole cake (i.e., $W_i = [0, 1]$), then $u_i([0,1]) = 1$ and $u_i(A_i) \le \alpha$, so no mechanism can achieve egalitarian ratio larger than $\alpha$.
The tightness follows from a mechanism that, given any instance, divides the cake into intervals using all breakpoints and chooses an (arbitrary) $\alpha$ fraction from each interval---this results in $u_i(A_i)=\alpha\cdot u_i([0,1])$ for all $i$.
\end{proof}

Our next result gives the precise egalitarian ratio of the leximin solution.

\begin{theorem}
\label{thm:leximin-ratio}
For all $n\ge 1$ and
$\alpha\in(0,1)$,
\[
\emph{Egal-ratio}_{n,\alpha}(\emph{leximin}) = \frac{\alpha}{n-(n-1)\alpha}.
\]
\end{theorem}

\begin{proof}
For the upper bound, consider the instance $R$ with $W_i = [(i-1)\alpha/n,i\alpha/n]$ for $i=1,\dots,n-1$, and $W_n = [(n-1)\alpha/n,1]$.
Every leximin allocation $A$ gives a desired cake of length $\alpha/n$ to every agent, so
\begin{align*}
\text{Egal-ratio-alloc}_R(A)
&\le \frac{u_n(A)}{u_n([0,1])} = \frac{\alpha/n}{1-\frac{(n-1)\alpha}{n}}= \frac{\alpha}{n-(n-1)\alpha}.
\end{align*}

We now prove the lower bound.
Since the mechanism can allocate length $\alpha$ of the cake and there are $n$ agents, it can give every agent~$i$ a utility of at least $\min\{\alpha/n, \ell(W_i)\}$.
Hence, all leximin allocations give each agent~$i$ at least this much utility.
If an agent has $\ell(W_i)\le 1 - (n-1)\alpha/n$, the utility ratio for this agent is at least $\frac{\alpha/n}{1-(n-1)\alpha/n} = \frac{\alpha}{n-(n-1)\alpha}$.
Else, suppose that  $\ell(W_i) = 1 - (n-1)\alpha/n + x$ for some $x > 0$.
In this case, no matter how the mechanism allocates length $\alpha$ of the cake, the utility of this agent is at least
\[
\alpha - \left(1 - \left(1 - \frac{(n-1)\alpha}{n} + x\right)\right) = \frac{\alpha}{n} + x.
\]
Hence, the utility ratio of this agent is at least
\[
\frac{\alpha/n + x}{1 - (n-1)\alpha/n + x} \ge \frac{\alpha}{n-(n-1)\alpha},
\]
where the inequality follows from the fact that the expression on the left-hand side is non-decreasing for $x\in[0,\infty)$.
\end{proof}

Theorem~\ref{thm:leximin-ratio} shows that the leximin solution achieves a non-trivial egalitarian ratio.
However, it is unclear how good this ratio is compared to that of other truthful mechanisms.
We will therefore show that the solution attains the highest possible ratio among all truthful mechanisms satisfying a natural condition.
Given a vector of piecewise uniform density functions $\mathbf{f} = (f_1,\dots,f_n)$, let $L_\mathbf{f}$ be a vector with $2^n$ components such that each component represents a distinct subset of agents and the value of the component is the length of the piece desired by exactly that subset of agents (and not by any agent outside the subset).

\begin{example}
Consider the instance in Example~\ref{ex:illustration}.
The corresponding $L_\mathbf{f}$ of this instance is $(1/8, 1/4, 3/8, 1/4)$,
where the components correspond to the lengths of the pieces desired by exactly the set of agents $\emptyset$, $\{1\}$, $\{2\}$, and $\{1, 2\}$, respectively.
\end{example}

\begin{definition}[Position obliviousness]
\label{def:position-oblivious}
A mechanism $\mech$ is \emph{position oblivious} if the following holds:

Let $\mathbf{f}$ and $\mathbf{f}'$ be any vectors of density functions such that $L_\mathbf{f} = L_{\mathbf{f}'}$, and let $R$ and $R'$ be instances represented by these respective vectors and a given parameter $\alpha$.
If $\mech(R) = (A,A_1,\dots,A_n)$ and $\mech(R') = (A',A_1',\dots,A_n')$, then $u_i(A_i) = u_i'(A_i')$ for every $i\in N$.
\end{definition}

Position obliviousness has previously been studied by \citet{BeiHuSu20}.
Intuitively, for a position oblivious mechanism, the utility of an agent depends only on the lengths of the pieces desired by various subsets of agents and not on the positions of these pieces.
It follows directly from the definition that the leximin solution is position oblivious.\footnote{\citet{BeiChHu17} considered a slightly stronger version of position obliviousness, which the leximin solution also satisfies.}

\begin{theorem}
\label{thm:impossibility}
Let $\mech$ be a truthful and position oblivious mechanism.
Then, for all $n\ge 1$ and $\alpha\in(0,1)$,
\[
\emph{Egal-ratio}_{n,\alpha}(\mech) \le \frac{\alpha}{n-(n-1)\alpha}.
\]
\end{theorem}
\begin{proof}
Assume for the sake of contradiction that there exists a truthful and position oblivious mechanism $\mech$ with $\egalratio = \frac{\alpha}{n - (n-1)\alpha} + \delta$ for some $\delta > 0$.
For each $i \in N$, let $C_i$ be a piece of length $\ell(C_i) = \alpha/n + \varepsilon$ such that $C_i \cap C_j = \emptyset$ for every pair $i, j \in N$, where $\varepsilon > 0$ is such that
\[
\varepsilon < \min\left\{\frac{1-\alpha}{n},\frac{\delta (n - (n-1)\alpha)^2}{n (n-1) (\alpha + \delta (n - (n-1)\alpha))}\right\}.
\]

Consider an instance $R$ where $W_i = C_i$ for all $i \in N$.
Since $\mech$ can allocate length at most $\alpha$ of the cake, it must return an allocation  for which some agent receives utility at most $\alpha / n$.
Assume without loss of generality that $\mech$ returns an allocation $\alloc$ with $u_1(A_1) \leq \alpha / n$.

Next, consider an instance $R'$ where $W_i' = C_i$ for all $i \in N \setminus \{1\}$ and $W_1' = [0, 1] \setminus \bigcup_{i \in N \setminus \{1\}} C_i$.
(We use the notation $W_i'$ for instance $R'$ to distinguish from $W_i$ for instance $R$.)
For this instance, we have $\ell(W_1') = 1 - (n-1) \cdot (\alpha / n + \varepsilon)$.
Let $\alloc' = \mech(R')$, and let $Y = A_1' \cap W_1'$.
By the definition of egalitarian ratio, we have
$u_1(A_1') / u_1([0, 1]) \geq \egalratio$, that is,
\begin{align*}
\ell(Y)
&\geq \egalratio \cdot \ell(W_1') = \left(\frac{\alpha}{n - (n-1)\alpha} + \delta\right) \cdot \left(1 - (n-1) \cdot \left(\frac{\alpha}{n} + \varepsilon\right)\right),
\end{align*}
which is greater than $\alpha/n$ by our choice of $\varepsilon$.

Finally, consider an instance $R''$ where $W_i''=C_i$ for all $i \in N \setminus \{1\}$, while $W_1''$ is a subset of $[0, 1] \setminus \bigcup_{i \in N \setminus \{1\}} C_i$ of length $\ell(W_1'') = \alpha / n + \varepsilon$ such that $\ell(W_1''\cap Y) > \alpha/n$.
Since $\mech$ is position oblivious, by comparing instances $R''$ with $R$, agent~$1$ must also get a utility of at most $\alpha / n$ in instance $R''$.
However, if the agent reports $[0, 1] \setminus \bigcup_{i \in N \setminus \{1\}} C_i$ as in $R'$, she gets a utility of $\ell(W_1''\cap Y) > \alpha/n$.
This means that $\mech$ is not truthful and yields the desired contradiction.
\end{proof}

Comparing this ratio with highest possible ratio of $\alpha$ without the truthfulness condition (Proposition~\ref{prop:egal-ratio-upper}),\footnote{Note that the mechanism that achieves egalitarian ratio $\alpha$ in Proposition~\ref{prop:egal-ratio-upper} satisfies position obliviousness.} one can see that adding the truthfulness requirement incurs a (multiplicative) ``price'' of $n - (n-1)\alpha$ on the best egalitarian ratio. This price can be as large as $n$ when $\alpha$ is close to $0$, and decreases to $1$ as $\alpha$ approaches~$1$.

\section{Maximum Nash Welfare}
\label{sec:MNW}

In this section, we address the MNW solution.
We start by showing that like the leximin solution (Proposition~\ref{prop:leximin-unique}), the utility that each agent receives is the same in all MNW allocations; this renders the tie-breaking issue insignificant.

\begin{proposition}
\label{prop:MNW-unique}
Given any instance, for each agent~$i$, the utility that $i$ receives is the same in all MNW allocations.
\end{proposition}

\begin{proof}
We proceed in a similar manner as in the proof of Proposition~\ref{prop:leximin-unique}.
Assume for contradiction that two MNW allocations, $A$ and $A'$, give some agent different utilities.
Since the utility of every agent in an MNW allocation is strictly positive, we have $u_i(A), u_i(A') > 0$ for all $i\in N$.
Let $A''$ be an allocation such that for each pair of consecutive breakpoints, the amount of cake between those breakpoints included in $A''$ is the average of the corresponding amounts for $A$ and $A'$.
By linearity, $A''$ is a feasible allocation, and $u_i(A'') = \frac{1}{2}(u_i(A) + u_i(A'))$ for every $i\in N$.

Recall that by the arithmetic-geometric mean inequality, it holds that $\frac{x+y}{2}\ge\sqrt{xy}$ for all positive real numbers $x,y$, with equality if and only if $x = y$.
We therefore have
\begin{align*}
\prod_{i\in N}u_i(A'')
&= \prod_{i\in N}\left(\frac{1}{2}(u_i(A) + u_i(A'))\right) \\
&> \prod_{i\in N}\sqrt{u_i(A)\cdot u_i(A')} = \sqrt{\prod_{i\in N}u_i(A)}\cdot\sqrt{\prod_{i\in N}u_i(A')},
\end{align*}
where the inequality is strict because $u_i(A)\ne u_i(A')$ for at least one $i$.
Since $\prod_{i\in N}u_i(A) = \prod_{i\in N}u_i(A')$, this implies that $A''$ has a higher Nash welfare than both $A$ and $A'$, yielding the desired contradiction.
\end{proof}

In the case of two agents, we show that MNW and leximin are in fact equivalent.
The high-level idea is that both solutions can be obtained via the following process: First, select portions of the cake desired by both agents.
If the quota $\alpha$ has not been reached, let the agents `eat' their desired piece using the same speed, until either (i) one of the agents has no more desired cake, in which case we let the other agent continue eating, or (ii) we run out of quota.

\begin{theorem}
\label{thm:MNW-leximin-equivalence}
Consider an instance with two agents.
Any leximin allocation is an MNW allocation, and vice versa.
\end{theorem}

\begin{proof}
Fix an instance with two agents, and let $X = W_1\cap W_2$ and $x=\ell(X)$.
If $x \geq \alpha$, then an allocation $A$ is leximin if and only if $A\subseteq X$, and the same holds for MNW.
Similarly, if $\ell(W_1 \cup W_2) \leq \alpha$, the relevant condition for both leximin and MNW is $W_1\cup W_2\subseteq A$.

Assume now that $x < \alpha < \ell(W_1 \cup W_2)$.
Since both the leximin and MNW solutions satisfy Pareto optimality, we must have $\ell(A) = \alpha$ and $X \subseteq A$ in any leximin or MNW allocation $A$.
In other words, the entire intersection of length $x$ must be allocated, along with a further length $\alpha-x$ of the cake.
Let $\Delta_1 = W_1\setminus W_2$ and $\Delta_2 = W_2\setminus W_1$, and consider two cases.

\underline{Case 1}: $\min\{\ell(\Delta_1), \ell(\Delta_2)\} \ge (\alpha-x)/2$.
In this case, for both leximin and MNW, the length $\alpha-x$ must be split equally between $\Delta_1$ and $\Delta_2$---otherwise the allocation can be improved with respect to both the leximin ordering and the Nash welfare by splitting the length equally.
Conversely, any allocation that splits the length $\alpha-x$ equally between $\Delta_1$ and $\Delta_2$ is both leximin and MNW.

\underline{Case 2}: $\min\{\ell(\Delta_1), \ell(\Delta_2)\} < (\alpha-x)/2$.
Assume without loss of generality that $\ell(\Delta_1) < (\alpha-x)/2$.
Since
$$\ell(\Delta_1) + \ell(\Delta_2) = \ell(W_1\cup W_2) - x > \alpha-x,$$
we have $\ell(\Delta_2) > (\alpha-x)/2$.
In this case, the entire $\Delta_1$ must be allocated---otherwise the allocation $A$ can be improved with respect to both the leximin ordering and the Nash welfare by allocating $\varepsilon$ more of $\Delta_1$ and $\varepsilon$ less of $\Delta_2$, for any $0 < \varepsilon < \ell(\Delta_1\setminus A)$.
Conversely, any allocation that allocates the entire $\Delta_1$ and length $\alpha-x-\ell(\Delta_1)$ of $\Delta_2$ is both leximin and MNW.

The desired conclusion follows from the two cases.
\end{proof}

Theorems~\ref{thm:leximin-truthful} and \ref{thm:MNW-leximin-equivalence} together imply the following:

\begin{corollary}
\label{cor:MNW-truthful-two}
For two agents and arbitrary tie-breaking, the MNW solution is truthful.
\end{corollary}

When $n\ge 3$, the two mechanisms are no longer equivalent.
This can be seen from the instance with $W_1=[0,1/2]$ and $W_i=[1/2,1]$ for all $2\le i\le n$, and $\alpha = 1/2$.
The leximin solution selects length $1/4$ from each half of the cake, while MNW selects length $\frac{1}{2n}$ from the first half and $\frac{n-1}{2n}$ from the second half.
For our main result of this section, we demonstrate that the MNW solution is not truthful even when an agent is only allowed to report a subset of her true desired piece---as discussed in Section~\ref{sec:our-results}, this strengthens the non-truthfulness result of \citet{AzizBoMo20} where the manipulation is not of this simple nature.
In particular, we construct an instance with six agents such that one of the agents can obtain a higher utility by reporting a subset of her actual desired piece.

\begin{theorem}
\label{thm:MNW-negative}
The MNW solution is not truthful under subset reporting regardless of tie-breaking.
\end{theorem}

\begin{proof}
Assume for convenience that the cake is represented by the interval $[0,8]$; this can be trivially scaled back down to $[0,1]$.
In our original instance, there are six agents whose utility functions are given as follows:
\begin{align*}
W_1 = [0,1]\cup[2,8],& \quad W_2 = [0,1]\cup [2,5], \\
W_3 = [0,1]\cup [5,8],& \quad W_4 = [1,3]\cup [5,6], \\
W_5 = [1,2]\cup [3,4]\cup [6,7],& \quad W_6 = [1,2]\cup [4,5]\cup [7,8],
\end{align*}
and let $\alpha = 2$.
See Figure~\ref{fig:MNW}.

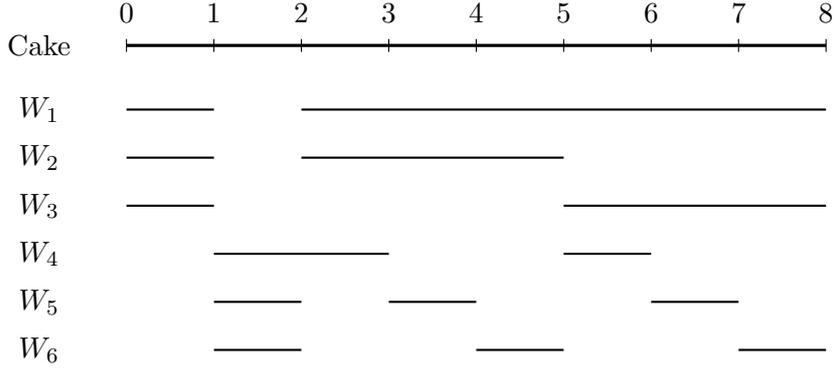
\begin{figure}
\centering
\begin{tikzpicture}[xscale=1.15,yscale=0.85]
\node at (2,5.5) {$0$};
\node at (3,5.5) {$1$};
\node at (4,5.5) {$2$};
\node at (5,5.5) {$3$};
\node at (6,5.5) {$4$};
\node at (7,5.5) {$5$};
\node at (8,5.5) {$6$};
\node at (9,5.5) {$7$};
\node at (10,5.5) {$8$};
\draw[very thick] (2,5) -- (10,5);
\draw (2,4.9) -- (2,5.1);
\draw (3,4.9) -- (3,5.1);
\draw (4,4.9) -- (4,5.1);
\draw (5,4.9) -- (5,5.1);
\draw (6,4.9) -- (6,5.1);
\draw (7,4.9) -- (7,5.1);
\draw (8,4.9) -- (8,5.1);
\draw (9,4.9) -- (9,5.1);
\draw (10,4.9) -- (10,5.1);
\node at (1,5) {Cake};
\draw[thick] (2,4) -- (3,4);
\draw[thick] (4,4) -- (10,4);
\node at (1,4) {$W_1$};
\draw[thick] (2,3.25) -- (3,3.25);
\draw[thick] (4,3.25) -- (7,3.25);
\node at (1,3.25) {$W_2$};
\draw[thick] (2,2.5) -- (3,2.5);
\draw[thick] (7,2.5) -- (10,2.5);
\node at (1,2.5) {$W_3$};
\draw[thick] (3,1.75) -- (5,1.75);
\draw[thick] (7,1.75) -- (8,1.75);
\node at (1,1.75) {$W_4$};
\draw[thick] (3,1) -- (4,1);
\draw[thick] (5,1) -- (6,1);
\draw[thick] (8,1) -- (9,1);
\node at (1,1) {$W_5$};
\draw[thick] (3,0.25) -- (4,0.25);
\draw[thick] (6,0.25) -- (7,0.25);
\draw[thick] (9,0.25) -- (10,0.25);
\node at (1,0.25) {$W_6$};
\end{tikzpicture}
\caption{The original instance in the proof of Theorem~\ref{thm:MNW-negative}.
}
\label{fig:MNW}
\end{figure}

First, observe that in this instance, every (non-integer) point is valued by exactly three agents.
Hence, for any subset~$A$ of the cake with $\ell(A)\le 2$, we have $\sum_{i=1}^6 u_i(A)\le 6$.
By the inequality of arithmetic and geometric means (AM-GM), it holds that $\prod_{i=1}^6 u_i(A) \le 1$.
By choosing $A=[0,2]$, we obtain $u_i(A)= 1$ for each $i$, so this choice of $A$ maximizes the Nash welfare as $\ell(A) = 2$ and $\prod_{i=1}^{6} u_i(A_i) = 1$, and gives agent~$1$ a utility of~$1$.
By Proposition~\ref{prop:MNW-unique}, agent~$1$ receives utility $1$ in every MNW allocation.

Next, consider a modified instance where agent~$1$ reports $W_1 = [2,8]$, which is a strict subset of the true $W_1$.
Consider an MNW allocation $A$ for this instance, and let $x:=\ell(A\cap [2,8])$, $y:= \ell(A\cap[1,2])$, and $z:=\ell(A\cap[0,1]) $, so $x+y+z\le 2$.
Let $A'$ be an allocation such that $\ell(A'\cap[0,1]) = 2-x-y\ge z$, $\ell(A'\cap[1,2]) = y$, and $|A'\cap [j,j+1]| = x/6$ for $j\in\{2,3,\dots,7\}$.
Notice that $\ell(A') = 2$, so $A'$ is a feasible allocation.
We claim that $\prod_{i=1}^6 u_i(A') \ge \prod_{i=1}^6 u_i(A)$.
Indeed, letting $\tau := \ell(A\cap [2,5])$, we have
\begin{align*}
u_2(A)\cdot u_3(A)
&= (z + \tau)(z + (x- \tau)) \le \left(z + \frac{x}{2}\right)^2 \le u_2(A')\cdot u_3(A'),
\end{align*}
where the first inequality follows from the AM-GM inequality.
Similarly, letting $\theta := \ell(A\cap ([2,3]\cup[5,6]))$ and $\rho := \ell(A\cap ([3,4]\cup[6,7]))$, it holds that
\begin{align*}
u_4(A)\cdot u_5(A)\cdot u_6(A) 
&= (y + \theta)(y + \rho)(y + (x-\theta-\rho)) \\
&\le \left(y + \frac{x}{3}\right)^3 \\ &= u_4(A')\cdot u_5(A')\cdot u_6(A').
\end{align*}
Moreover, since $u_1(A) = u_1(A') = x$, it follows that $\prod_{i=1}^6 u_i(A') \ge \prod_{i=1}^6 u_i(A)$, as claimed.
This means that $A'$ is also an MNW allocation.
The Nash welfare of $A'$ is
\[
\prod_{i=1}^6 u_i(A') = x\left(2-\frac{x}{2}-y\right)^2\left(y+\frac{x}{3}\right)^3.
\]
In order to show that MNW is not truthful regardless of tie-breaking, by Proposition~\ref{prop:MNW-unique}, it suffices to show that the maximum of this expression in the domain $x,y\ge 0, x+y\in[0,2]$ is attained when $x > 1$, since this would imply that agent~$1$ has a profitable deviation.

Let $g(x,y) := x\left(2-\frac{x}{2}-y\right)^2\left(y+\frac{x}{3}\right)^3$, where $x,y\ge 0$ and $x+y\le 2$.
We have $g(1.5,0.5) = 0.84375$.
Now, from the AM-GM inequality,
\begin{align*}
\frac{9}{4}\cdot g(x,y)
&= x\left(3-\frac{3x}{4}-\frac{3y}{2}\right)^2\left(y+\frac{x}{3}\right)^3 \\
&\le x\left(\frac{2(3-3x/4-3y/2) + 3(y+x/3)}{5}\right)^5 = x\left(\frac{6-x/2}{5}\right)^5.
\end{align*}
The derivative of the last expression is
$
\left(\frac{6-3x}{5}\right)\left(\frac{6-x/2}{5}\right)^4
$,
which is nonnegative for $0\le x\le 2$.
This means that for $x\le 1$, we have
\[
\frac{9}{4}\cdot g(x,y) \le 1\cdot\left(\frac{6-1/2}{5}\right)^5=\left(\frac{11}{10}\right)^5,
\]
so $g(x,y)\le \frac{4}{9}\cdot (1.1)^5 < 0.72 < g(1.5,0.5)$.
It follows that the maximum of $g(x,y)$ is attained when $x > 1$, as desired.
\end{proof}

We remark here that even if we allow the MNW solution to choose any $A_i$ such that $A\cap W_i\subseteq A_i\subseteq A$ instead of always choosing $A_i = A\cap W_i$ (that is, the mechanism may give agent~$i$ some parts of $A$ that she does not value, along with all parts of $A$ that she values), our example in Theorem~\ref{thm:MNW-negative} still shows that any resulting mechanism is not truthful under subset reporting.

From the fairness perspective, we show later in \Cref{thm:MNW-ratio} that the MNW solution achieves the same egalitarian ratio as the leximin solution.
However, the fact that the MNW solution is prone to a particularly simple form of manipulation makes it an unsuitable choice when truthfulness is essential.

\section{Impossibility Result Without Blocking}
\label{sec:no-blocking}

As we have so far assumed that mechanisms can block agents from accessing certain parts of the resource, an interesting question is what guarantees the mechanisms can achieve without the ability to block.
Indeed, while blocking can be easily implemented in our introductory applications by restricting access to the sports facility or files in a cache memory, it may be harder or more costly in other situations, e.g., cleaning streets or constructing public parks.
In this section, we consider mechanisms without the blocking ability.
When no blocking is allowed, given an input instance, a mechanism $\mech$ simply chooses a piece of cake~$A$ with $\ell(A) \le \alpha$, and each agent $i$ receives a utility of $u_i(A) = \ell(A \cap W_i)$.

First, we observe that while the leximin solution is truthful if it has the ability to block (Theorem~\ref{thm:leximin-truthful}), this is no longer the case in the absence of blocking.

\begin{example}[Leximin is not truthful without blocking]
\label{ex:leximin-w.o-blocking}
Let $\alpha = 1/2$.
First, consider an instance $R$ with two agents whose utility functions are given as follows:
\[
W_1 = [0, 1/2], \quad W_2 = [1/2, 1].
\]
Assume without loss of generality that the tie-breaking rule chooses $A = [1/4, 3/4]$.
Next, consider an instance $R'$ with the following utility functions:
\[
W_1 = [0, 3/4], \quad W_2 = [1/2, 1].
\]
Agent~$1$ receives a utility of $3/8$ in every leximin allocation for $R'$.
However, if agent~$1$ misreports that $W_1 = [0, 1/2]$, the instance becomes the same as $R$, and agent~$1$ receives a utility of $1/2$ from the allocation $A$.
\end{example}

Our main result of this section shows that Example~\ref{ex:leximin-w.o-blocking} is in fact not a coincidence.

\begin{theorem}\label{thm:impossibility-no-blocking}
Without blocking, for every $\alpha \in (0, 1)$, no truthful, Pareto optimal, and position oblivious mechanism can achieve a positive egalitarian ratio even in the case of two agents.
\end{theorem}

\begin{proof}
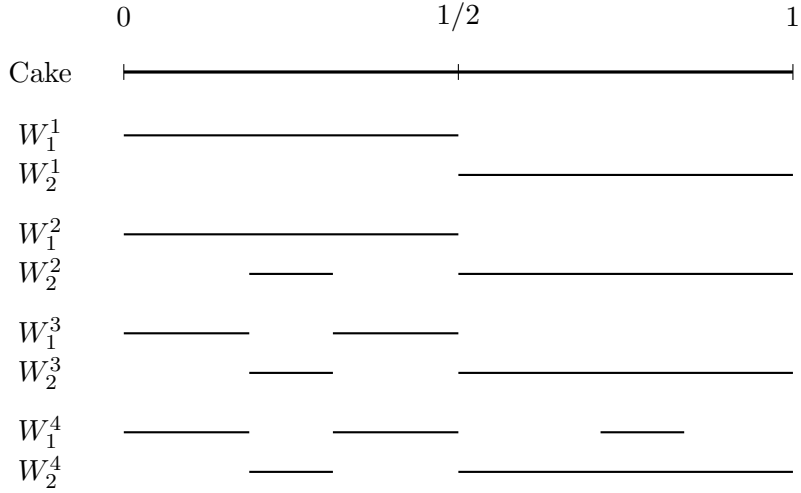
\begin{figure}
\centering
\begin{tikzpicture}[xscale=1.1,yscale=1.05]
\node at (2,5.5) {$0$};
\node at (6,5.5) {$1/2$};
\node at (10,5.5) {$1$};
\draw[very thick] (2,4.8) -- (10,4.8);
\draw (2,4.7) -- (2,4.9);
\draw (6,4.7) -- (6,4.9);
\draw (10,4.7) -- (10,4.9);
\node at (1,4.8) {Cake};
\draw[thick] (2,4) -- (6,4);
\node at (1,4) {$W^1_1$};
\draw[thick] (6,3.5) -- (10,3.5);
\node at (1,3.5) {$W^1_2$};
\draw[thick] (2,2.75) -- (6,2.75);
\node at (1,2.75) {$W^2_1$};
\draw[thick] (3.5,2.25) -- (4.5,2.25);
\draw[thick] (6,2.25) -- (10,2.25);
\node at (1,2.25) {$W^2_2$};
\draw[thick] (2,1.5) -- (3.5,1.5);
\draw[thick] (4.5,1.5) -- (6,1.5);
\node at (1,1.5) {$W^3_1$};
\draw[thick] (3.5,1) -- (4.5,1);
\draw[thick] (6,1) -- (10,1);
\node at (1,1) {$W^3_2$};
\draw[thick] (2,0.25) -- (3.5,0.25);
\draw[thick] (4.5,0.25) -- (6,0.25);
\draw[thick] (7.7,0.25) -- (8.7,0.25);
\node at (1,0.25) {$W^4_1$};
\draw[thick] (3.5,-0.25) -- (4.5,-0.25);
\draw[thick] (6,-0.25) -- (10,-0.25);
\node at (1,-0.25) {$W^4_2$};
\end{tikzpicture}
\caption{Example instances in the proof of Theorem~\ref{thm:impossibility-no-blocking}.
}
\label{fig:impossibility-no-blocking}
\end{figure}

We assume for contradiction that there exists some $\alpha \in (0, 1)$ and a truthful, Pareto optimal, and position oblivious mechanism $\mech$ with $\text{Egal-ratio}_{2,\alpha}(\mech) > 0$.
We consider a sequence of instances with two agents, which we illustrate in Figure~\ref{fig:impossibility-no-blocking}.
In the following, the superscripts denote the indices of the instances.
In all of the instances that we consider, every part of the cake is desired by at least one agent, so Pareto optimality implies that $\mech$ must allocate exactly $\alpha$ of the cake.
\begin{itemize}
\item Instance $R^1$: $W^1_1 = [0, 0.5], W^1_2 = [0.5, 1]$. \\
Let $\mech(R^1) = A^1$.
Because $\alpha < 1$,  at least one of the agents will not obtain her maximum utility of $0.5$.
Assume without loss of generality that $\ell(W^1_1 \cap A^1) = x < 0.5$; in other words, $W^1_1 \setminus A^1$ is nonempty.
Since $\mech$ has a positive egalitarian ratio, it must hold that $0 < x < \alpha$.

\item Instance $R^2$: $W^2_1 = [0, 0.5], W^2_2 = A^1 \cup [0.5, 1]$. \\
Let $\mech(R^2) = A^2$.
We must have $A^2 \subseteq W^2_2$; in other words, agent~$2$ will receive utility $\alpha$.
This is because otherwise, agent~$2$ can benefit by reporting ${W^2_2}' = [0.5, 1]$ and the instance becomes $R^1$, in which case agent~$2$ will receive utility $\alpha$ from the output allocation $A^1$.
Note that because $A^2$ is contained entirely in $W^2_2$, we still have $\ell(W^2_1 \cap A^2) \leq x$.

\item Instance $R^3$: $W^3_1 = [0, 0.5] \setminus A^1, W^3_2 = A^1 \cup [0.5, 1]$. \\
Let $\mech(R^3) = A^3$.
By the positive egalitarian ratio, we have $\ell(W^3_1 \cap A^3) = y > 0$.

\item Instance $R^4$: $W^4_1 = \left([0, 0.5] \setminus A^1\right) \cup B, W^4_2 = A^1 \cup [0.5, 1]$, where $B$ is an interval of length $x$ contained in $W^3_2$ with the largest intersection with $A^3$.
That is,
\begin{itemize}
\item if $\ell(W^4_2 \cap A^3) \geq x$, let $B$ be any subset of $W^4_2 \cap A^3$ of length $x$;
\item if $\ell(W^4_2 \cap A^3) < x$, let $B$ be any interval of length~$x$ that contains $W^4_2 \cap A^3$.
\end{itemize}
Let $\mech(R^4) = A^4$.
In this instance, we must have $u_1(A^4) > x$.
This is because otherwise, agent~$1$ can benefit by reporting ${W^4_1}' = [0, 0.5] \setminus A^1$ and the instance becomes $R^3$, in which case agent~$1$ will obtain a utility of $x + y$ (when $\ell(W^4_2 \cap A^3) \geq x$) or a utility of $\alpha$ (when $\ell(W^4_2 \cap A^3) < x$).
In both cases this value is strictly larger than $x$.
\end{itemize}
Finally, observe that instances $R^2$ and $R^4$ have the same $L_\mathbf{f}$ vector.
In particular, we have
$\ell(W^2_1) = \ell(W^4_1) = 1/2$, $\ell(W^2_2) = \ell(W^4_2) = 1/2+x$, and $\ell(W^2_1 \cap W^2_2) = \ell(W^4_1 \cap W^4_2) = x$.
This means that each agent should receive the same utility in these two instances from our position oblivious mechanism $\mech$.
However, agent~$1$ receives utility at most $x$ in $R_2$ and utility strictly larger than $x$ in $R^4$.
We have reached a contradiction.
\end{proof}

\section{Conclusion and Future Work}
\label{sec:discussion}

In this paper, we have studied truthful and fair mechanisms in the cake sharing setting where all agents share the same subset of a heterogeneous divisible resource.
Our results establish the leximin solution as an attractive mechanism due to its truthfulness and its optimal egalitarian ratio among all truthful and position oblivious mechanisms.
On the other hand, we constructed an intricate example showing that the maximum Nash welfare (MNW) solution, which often exhibits desirable properties in other settings, fails to yield truthfulness in cake sharing even when the agents are restricted to subset reporting.
Moreover, we showed that in the absence of blocking, no truthful, Pareto optimal, and position oblivious mechanism can achieve a positive egalitarian ratio---in particular, this implies that the leximin solution is \emph{not} truthful without blocking.
An intriguing question is whether the impossibility still holds if we remove Pareto optimality or position obliviousness (or both), or whether there is a truthful mechanism that attains a non-trivial fairness guarantee even when blocking is not allowed.
Besides the properties that we have investigated, one could consider other desirable properties, for example those based on justified representation notions in multiwinner voting \citep{AzizBrCo17,Sanchez-FernandezElLa17}.
In Appendix~\ref{app:JR}, we discuss how such notions can be adapted to cake sharing, and present results on leximin and MNW with respect to these notions.

In future research, it would be interesting to extend our cake sharing model to capture other practical scenarios.
One natural direction is to allow agents to have more complex preferences beyond piecewise uniform utilities.
The first step in this direction would be to consider \emph{piecewise constant} utilities, where an agent's density function is constant over subintervals of the cake.
Another extension is to allow non-uniform costs over the cake---this models, for example, the fact that reserving a sports facility or a conference room can be more expensive during peak periods.
In Appendix~\ref{app:non-uniform-costs}, we show that a natural generalization of the leximin solution is still truthful and achieves the optimal egalitarian ratio for piecewise constant cost functions.
Furthermore, as in cake cutting, it may be fruitful to consider scenarios in which there are constraints on the shared cake \citep{Suksompong21}.\footnote{A desirable property in certain applications is \emph{contiguity}---for example, a contiguous time slot is often more useful than a union of disconnected slots.
However, note that if contiguity is imposed, we may not be able to avoid leaving some agents empty-handed.
A trivial example is when one agent only values a small piece of cake at the left end while another agent only values one at the right end.
In this case, unless $\alpha$ is very close to $1$, one of the agents will necessarily receive utility $0$.}
Other questions addressed in cake cutting, such as the \emph{price of truthfulness} and the \emph{price of fairness}---that is, the loss of social welfare due to truthfulness and fairness, respectively  \citep{CaragiannisKaKa12,MayaNi12}---are equally relevant and worthy of exploration in cake sharing as well.

\section*{Acknowledgments}

This work was partially supported by  the Ministry of Education, Singapore, under its Academic Research Fund Tier~1 (RG23/20), and by an NUS Start-up Grant.
The work was mostly done while the second author was at Nanyang Technological University and the National University of Singapore.
We would like to thank Haris Aziz for suggesting \Cref{def:AFS} and \Cref{thm:MNW-AFS}, Ilya Bogdanov for help with the proof of Theorem~\ref{thm:MNW-negative}, and the anonymous reviewers of AAAI 2022 for their valuable feedback.

\bibliographystyle{plainnat}
\bibliography{main}

\appendix

\section{Justified Representation Notions}
\label{app:JR}

In multiwinner voting, where the goal is to choose a certain number of candidates from the available candidates \citep{FaliszewskiSkSl17,LacknerSk21}, a desideratum that has received significant attention in recent years is \emph{justified representation (JR)} \citep{AzizBrCo17}.
If there are $n$ agents with approval preferences over the candidates and $k$ candidates are to be chosen, JR demands that whenever at least $n/k$ agents approve a common candidate, at least one of these agents is represented in the selected set of candidates (usually referred to as a \emph{committee}).
Two important strengthenings of JR have been proposed.
\begin{itemize}
\item A committee is said to provide \emph{proportional justified representation (PJR)} \citep{Sanchez-FernandezElLa17} if for every positive integer $r$ and every set of agents $N^*\subseteq N$ such that $|N^*| \ge r\cdot n/k$ and the agents in $N^*$ commonly approve $r$ candidates, at least $r$ candidates in the committee are approved by at least one of these agents.
\item A committee is said to provide \emph{extended justified representation (EJR)} \citep{AzizBrCo17} if for every positive integer $r$ and every set of agents $N^*\subseteq N$ such that $|N^*| \ge r\cdot n/k$ and the agents in $N^*$ commonly approve $r$ candidates, some agent in $N^*$ has at least $r$ approved candidates in the committee.
\end{itemize}
It follows directly from the definitions that EJR implies PJR, which in turn implies JR.
\citet{AzizBrCo17} showed that a committee providing EJR (and therefore PJR and JR) always exists.

In this section, we discuss how these notions can be adapted to our cake sharing setting, and how leximin and MNW perform with respect to the resulting notions.
First, observe that since there is no discrete unit of candidate in our setting, JR does not have a clear analog in cake sharing.
By contrast, both PJR and EJR admit natural adaptations.
Since we do not address truthfulness in this section, we ignore the issue of blocking and assume that an allocation is simply a piece of cake $A$ with $\ell(A) \le \alpha$.

\begin{definition}
\label{def:PJR}
Given an instance, an allocation~$A$ with $\ell(A) \le \alpha$ is said to provide \emph{proportional justified representation (PJR)} if for every positive real number $t$ and every set of agents $N^*\subseteq N$ such that $|N^*| \geq t\cdot n/\alpha$ and $\ell\left(\bigcap_{i \in N^*} W_i\right) \geq t$, it holds that $\ell\left(A \cap \left( \bigcup_{i \in N^*} W_i \right)\right) \geq t$.
\end{definition}

\begin{definition}
\label{def:EJR}
Given an instance, an allocation~$A$ with $\ell(A) \le \alpha$ is said to provide \emph{extended justified representation (EJR)} if for every positive real number $t$ and every set of agents $N^*\subseteq N$ such that $|N^*| \geq t\cdot n/\alpha$ and $\ell\left(\bigcap_{i \in N^*} W_i\right) \geq t$, it holds that $\ell(A\cap W_{j}) \geq t$ for some~$j \in N^*$.
\end{definition}

As in the discrete setting, it is clear that EJR implies PJR.
A mechanism is said to satisfy PJR (resp., EJR) if it outputs an allocation that provides PJR (resp., EJR) for every instance.

Since leximin may ignore the preferences of large groups of agents in order to maximize the minimum utility, it can be easily shown to violate both axioms.

\begin{proposition}
The leximin solution satisfies neither PJR nor EJR.
\end{proposition}

\begin{proof}
Since PJR is implied by EJR, it suffices to prove the statement for PJR.
Let $n\ge 3$, and consider an instance with $W_1 = [0, 1/2]$ and $W_i = [1/2, 1]$ for all $2 \leq i \leq n$.
The leximin solution selects length $\alpha / 2$ from each half of the cake.
Choose $t$ such that
\[
\frac{\alpha}{2} < t < \min \left\{ \alpha \cdot \frac{n-1}{n}, \frac{1}{2} \right\};
\]
this choice is always possible as $\alpha\in(0,1)$.
Since $n - 1 \geq t \cdot n/\alpha$ and agents~$2,\dots,n$ commonly approve a piece of cake of length $1/2 \ge t$, PJR dictates that a piece of length at least~$t$ must be selected from~$[1/2, 1]$.
However, leximin selects length $\alpha/2 < t$ from this interval.
\end{proof}

Next, we show that MNW satisfies both of the representation notions.
This result can be seen as a continuous analog of \citet{AzizBrCo17}'s result that the \emph{Proportional Approval Voting (PAV)} rule satisfies EJR in the discrete setting.
Indeed, PAV assigns a utility of $1 + \frac{1}{2}+\dots + \frac{1}{j} \approx \ln j$ to an agent whenever the agent approves $j$ candidates in the committee, and chooses a committee that maximizes the sum of the agents' utilities.
Since maximizing a product is equivalent to maximizing the sum of the logarithms of its terms, PAV in multiwinner voting is closely related to MNW in cake sharing.

To aid the presentation of our proof, we introduce some definitions.

\begin{definition}
\label{def:Nash-marginal}
Given an instance and an allocation~$A$, for every point~$x$ that is not included in~$A$ and is not a breakpoint, we define its \emph{Nash addition marginal} as
\[
\text{NAM}_A(x) = \lim_{\varepsilon \to 0} \frac{\sum_{i \in N} \ln u_i(\allocAdd) - \sum_{i \in N} \ln u_i(A)}{\varepsilon},
\]
where~$\allocAdd$ is obtained by adding to $A$ a piece of cake of length~$\varepsilon$ adjacent to the point~$x$.
Analogously, for every point~$x$ that is included in~$A$ and is not a breakpoint, we define its \emph{Nash removal marginal} as
\[
\text{NRM}_A(x) = \lim_{\varepsilon \to 0} \frac{\sum_{i \in N} \ln u_i(A) - \sum_{i \in N} \ln u_i(\allocRemove)}{\varepsilon},
\]
where~$\allocRemove$ is obtained by removing from $A$ a piece of cake of length~$\varepsilon$ adjacent to the point $x$.
\end{definition} 

Intuitively, the Nash addition marginal is the rate of change of (the logarithm of) the Nash welfare when we add cake adjacent to a certain point.
We do not define the Nash addition marginal for breakpoints because adding cake to the left and to the right of a breakpoint may yield different marginals.
Similar statements hold for the Nash removal marginal.

\begin{theorem}\label{thm:MNW-EJR}
The MNW solution satisfies EJR and PJR.
\end{theorem}

\begin{proof}
Since EJR implies PJR, it suffices to establish the statement for EJR.
Suppose for contradiction that MNW violates EJR for some instance with parameter $\alpha$, and consider a value $t$ and a set of agents~$N^*$ with $|N^*| \geq t \cdot n/\alpha$ who commonly approve a piece of cake~$S$ with $\ell(S) \geq t$ such that Definition~\ref{def:EJR} is not satisfied.
Let~$A$ be the piece of cake chosen by MNW, and assume without loss of generality\footnote{If $\ell(A) < \alpha$, this means that $\ell\left(\bigcup_{i\in N}W_i\right) < \alpha$, and any allocation returned by the MNW solution trivially provides EJR.} that $\ell(A) = \alpha$.
Since Definition~\ref{def:EJR} is violated, we have $u_i(A)  < t$ for all $i\in N^*$.

Suppose that $A$ is a disjoint union of the intervals $I_1,\dots,I_k$, where each interval $I_j$ is valued either entirely or not at all by each agent.
We have that $\ell(S) \ge t$, $u_i(A) < t$ for every $i\in N^*$, and all agents in $N^*$ approve the entire piece of cake $S$.
In particular, not all of $S$ is contained in $A$.
Let $z$ be a point in $S\setminus A$ that is not a breakpoint.
The contribution of each agent~$i \in N^*$ to the Nash addition marginal at point~$z$ is $(\ln u_i(A))' = 1 / u_i(A)$.
Thus,
\[
\text{NAM}_A(z) \geq  |N^*| \cdot \frac{1}{\max_{i \in N^*} u_i(A)} \geq t \cdot \frac{n}{\alpha} \cdot \frac{1}{\max_{i \in N^*} u_i(A)} > t \cdot \frac{n}{\alpha} \cdot \frac{1}{t} = \frac{n}{\alpha}.
\]
On the other hand, we have
\begin{align*}
\int_A \text{NRM}_A(x) \diff x 
&= \sum_{j = 1}^{k} \int_{I_j} \text{NRM}_A(x) \diff x \\
&= \sum_{j = 1}^{k} \int_{I_j} \left( \sum_{i \in N \colon I_j \subseteq W_i} \frac{1}{u_i(A)} \right) \diff x \\
&= \sum_{i \in N} \left( \int_{A\cap W_i} \frac{1}{u_i(A)} \diff x \right) = \sum_{i \in N} \frac{\ell(A\cap W_i)}{u_i(A)} = \sum_{i \in N} \frac{u_i(A)}{u_i(A)} = n;
\end{align*}
note that when we take integrals, we can safely ignore breakpoints because they form a set of measure zero.
Since $\ell(A) = \alpha$, there exists a point~$y\in A$ whose Nash removal marginal is at most $n/\alpha$.
In particular, the Nash addition marginal at $z$ is  strictly greater than the Nash removal marginal at~$y$.
Hence, by adding a piece of cake of sufficiently small length $\varepsilon > 0$ adjacent to $z$ and removing a piece of cake of the same length adjacent to $y$, we obtain an allocation with a larger Nash welfare than $A$.
This yields the desired contradiction.
\end{proof}

Theorem~\ref{thm:MNW-EJR} allows us to determine the egalitarian ratio of the MNW solution, which turns out to be the same as that of the leximin solution (\Cref{thm:leximin-ratio}).

\begin{theorem}
\label{thm:MNW-ratio}
For all $n\ge 1$ and
$\alpha\in(0,1)$,
\[
\emph{Egal-ratio}_{n,\alpha}(\emph{MNW}) = \frac{\alpha}{n-(n-1)\alpha}.
\]
\end{theorem}

\begin{proof}
For the upper bound, consider the same instance $R$ as in \Cref{thm:leximin-ratio}: $W_i = [(i-1)\alpha/n,i\alpha/n]$ for $i=1,\dots,n-1$, and $W_n = [(n-1)\alpha/n,1]$.
By the inequality of arithmetic and geometric means, every MNW allocation $A$ gives a desired cake of length $\alpha/n$ to every agent, so
\begin{align*}
\text{Egal-ratio-alloc}_R(A)
&\le \frac{u_n(A)}{u_n([0,1])} = \frac{\alpha/n}{1-\frac{(n-1)\alpha}{n}}= \frac{\alpha}{n-(n-1)\alpha}.
\end{align*}

For the lower bound, as in the proof of \Cref{thm:leximin-ratio}, it suffices to show that MNW guarantees each agent~$i \in N$ a utility of at least $\min \{\alpha/n, \ell(W_i)\}$.
To this end, we fix an agent~$i$ and let $t = \min \{\alpha/n, \ell(W_i)\}$.
We have $|\{i\}| = \frac{\alpha}{n} \cdot \frac{n}{\alpha} \ge t \cdot \frac{n}{\alpha}$ and $\ell(W_i) \ge t$.
Since the MNW solution satisfies EJR by \Cref{thm:MNW-EJR}, agent~$i$ must receive utility at least $t = \min \{\alpha/n, \ell(W_i)\}$, as desired.
\end{proof}

Next, we show that the MNW solution is the unique rule within a family of welfare-maximizer rules that provides either PJR or EJR.
\begin{definition}
Let $f:[0, 1] \to [-\infty, \infty)$ be a strictly increasing function which is differentiable in $(0, 1)$.
Given an instance, the \emph{$f$-welfare-maximizer rule} chooses a piece of cake $A$ with $\ell(A)\le\alpha$ whose \emph{$f$-welfare}, defined as the sum $\sum_{i \in N} f(u_i(A))$, is maximized. 
\end{definition}
Note that the $f$-welfare-maximizer rule with $f(x) = \ln x$ is equivalent to the MNW solution; the same is true for $f(x) = c\ln x + d$ for any real constants $c > 0$ and $d$.
For all of these functions $f$, it holds that $f'(x) = c/x$ for some constant $c$.

\begin{theorem}
\label{thm:welfare-maximizer}
Let $f:[0, 1] \to [-\infty, \infty)$ be a strictly increasing function which is differentiable in $(0, 1)$.
If the $f$-welfare-maximizer rule satisfies PJR, then there exists a constant $c$ such that $f'(x) = c/x$ for all $x\in (0,1)$.
The same statement holds for EJR.
\end{theorem}

\begin{proof}
Since PJR is implied by EJR, it suffices to prove the theorem for PJR.
Let $f$ be such that the $f$-welfare-maximizer rule satisfies PJR.

First, we show that for all $x,y > 0$ with $x + y < 1$, it holds that
\begin{equation}
\frac{f'(x)}{f'(y)} = \frac{y}{x}.  \label{eq:ratio}
\end{equation}
Suppose for contradiction that $\frac{f'(x)}{f'(y)} \ne \frac{y}{x}$ for some such $x,y$.
Assume that $\frac{f'(x)}{f'(y)} > \frac{y}{x}$; otherwise, we can switch the roles of $x$ and $y$.
Let $n_x$ and $n_y$ be positive integers such that
\[
\frac{f'(x)}{f'(y)} > \frac{n_y}{n_x} > \frac{y}{x}.
\]
We construct an instance as follows.
Let $\alpha = x + y < 1$.
The cake is a union of two disjoint intervals: $S_x$ with length~$\ell(S_x) = 1 - y > x$, and~$S_y$ with length $\ell(S_y) = y$.
The set of agents $N$ is composed of two disjoint sets~$N_x$ and~$N_y$ with $|N_x| = n_x$ and $|N_y| = n_y$ (so $n = n_x + n_y$) such that the agents in~$N_x$ (resp.,~$N_y$) only approve the piece of cake~$S_x$ (resp.,~$S_y$).
Since
\[
|N_y| = n_y > y \cdot \frac{n}{\alpha} = y \cdot \frac{n_x + n_y}{x + y}
\]
and all agents in $N_y$ approve a common piece of cake of length $y$, by the PJR condition, an allocation $A$ chosen by the $f$-welfare-maximizer rule must include the entire piece $S_y$.
In addition, as $\ell(S_x) > x$, $A$ cannot include the entire piece $S_x$.
However, since $n_x \cdot f'(x) > n_y \cdot f'(y)$, the $f$-welfare of $A$ can be improved by including an $\varepsilon$-length more of $S_x$ and an $\varepsilon$-length less of $S_y$, for some sufficiently small $\varepsilon > 0$.
This contradicts the assumption that $A$ maximizes the $f$-welfare among all allocations of length at most~$\alpha$, and therefore establishes \eqref{eq:ratio}.

Next, we extend \eqref{eq:ratio} by showing that $\frac{f'(x)}{f'(y)} = \frac{y}{x}$ for all $x,y\in(0,1)$.
Fix $x,y\in(0,1)$, and let $\varepsilon > 0$ be such that $x + \varepsilon < 1$ and $y + \varepsilon < 1$.
From \eqref{eq:ratio}, we have $\frac{f'(x)}{f'(\varepsilon)} = \frac{\varepsilon}{x}$ and $\frac{f'(y)}{f'(\varepsilon)} = \frac{\varepsilon}{y}$.
Dividing the first relation by the second yields $\frac{f'(x)}{f'(y)} = \frac{y}{x}$.

Finally, we show that there exists a constant $c$ such that $f'(x) = c/x$ for all $x\in(0,1)$.
Fix $b\in (0,1)$.
For any $x\in (0,1)$, we have $\frac{f'(x)}{f'(b)} = \frac{b}{x}$, and so
\[
f'(x) = \frac{b \cdot f'(b)}{x}.
\]
Thus, we can take $c = b \cdot f'(b)$.
This completes the proof of the theorem.
\end{proof}

To end this section, we strengthen \Cref{thm:MNW-EJR} by adapting a fairness notion due to \citet{AzizBoMo20}.

\begin{definition}
\label{def:AFS}
Given an instance, an allocation~$A$ with $\ell(A) \le \alpha$ is said to provide \emph{average fair share (AFS)} if for every positive real number $t$ and every set of agents $N^*\subseteq N$ such that $|N^*| \geq t\cdot n/\alpha$ and $\ell\left(\bigcap_{i \in N^*} W_i\right) \geq t$, it holds that $\sum_{i\in N^*} u_i(A) \geq |N^*|\cdot t$.
\end{definition}

Since $\sum_{i\in N^*} u_i(A) \geq |N^*|\cdot t$ implies that $u_j(A) \ge t$ for at least one $j\in N^*$, AFS is a strengthening of EJR.
As with EJR and PJR, we say that a mechanism satisfies AFS if it outputs an allocation that provides AFS for every instance.

\begin{theorem}\label{thm:MNW-AFS}
The MNW solution satisfies AFS.
\end{theorem}

\begin{proof}
The proof is very similar to that of \Cref{thm:MNW-EJR}.
Suppose for contradiction that MNW violates AFS for some instance with parameter $\alpha$, and consider a value $t$ and a set of agents~$N^*$ with $|N^*| \geq t \cdot n/\alpha$ who commonly approve a piece of cake~$S$ with $\ell(S) \geq t$ such that Definition~\ref{def:AFS} is not satisfied.
Let~$A$ be the piece of cake chosen by MNW, and assume without loss of generality that $\ell(A) = \alpha$.
Since Definition~\ref{def:AFS} is violated, we have $\sum_{i\in N^*}u_i(A)  < |N^*|\cdot t$.

Suppose that $A$ is a disjoint union of the intervals $I_1,\dots,I_k$, where each interval $I_j$ is valued either entirely or not at all by each agent.
We have that $\ell(S) \ge t$, $\sum_{i\in N^*}u_i(A)  < |N^*|\cdot t$, and all agents in $N^*$ approve the entire piece of cake $S$.
In particular, not all of~$S$ is contained in $A$.
Let $z$ be a point in $S\setminus A$ that is not a breakpoint.
The contribution of each agent~$i \in N^*$ to the Nash addition marginal at point~$z$ is $(\ln u_i(A))' = 1 / u_i(A)$.
Thus,
\[
\text{NAM}_A(z) \geq  \sum_{i\in N^*} \frac{1}{u_i(A)} \geq \frac{|N^*|^2}{\sum_{i\in N^*}u_i(A)} > \frac{|N^*|^2}{|N^*|\cdot t} = \frac{|N^*|}{t} \ge  \frac{n}{\alpha},
\]
where for the second inequality we apply the inequality of arithmetic and harmonic means.
On the other hand, as in the proof of \Cref{thm:MNW-EJR}, there exists a point~$y\in A$ whose Nash removal marginal is at most $n/\alpha$.
In particular, the Nash addition marginal at $z$ is  strictly greater than the Nash removal marginal at~$y$.
Hence, by adding a piece of cake of sufficiently small length $\varepsilon > 0$ adjacent to $z$ and removing a piece of cake of the same length adjacent to $y$, we obtain an allocation with a larger Nash welfare than~$A$.
This yields the desired contradiction.
\end{proof}

\section{Non-Uniform Costs}
\label{app:non-uniform-costs}

In this section, we consider an extension of our model where the cost of selecting the cake may be non-uniform.
Specifically, there is a (public) \emph{cost function} $c:[0,1]\rightarrow\mathbb{R}_{\ge 0}$, which captures the cost for different parts of the cake.
We assume that the cost function is piecewise constant, and without loss of generality that $\int_0^1 c\, dx  = 1$ (if the whole cake has cost $0$, the mechanism can simply always choose the whole cake).
Note that the main model of our paper corresponds to the cost function being the constant~$1$ over the entire cake.
We still consider piecewise uniform utility functions of the agents, and allow the mechanism to choose a piece of cake with cost at most a given parameter $\alpha\in(0,1)$.

We can generalize the leximin solution to this setting as follows.
First, consider all breakpoints of the cost function, where the breakpoints are defined in the same way as for the utility functions.
Then, for each piece of cake between two consecutive breakpoints, we choose a fraction of at most $\alpha$ of this cake by implementing the canonical leximin solution with the same $\alpha$.
The generalized leximin solution then returns the union of the chosen cake.
By linearity, the chosen cake has cost at most $\alpha$.

\begin{theorem}
For all $n\ge 1$ and $\alpha\in(0,1)$, when the cost function is piecewise constant, the generalized leximin solution is truthful and has egalitarian ratio $\frac{\alpha}{n-(n-1)\alpha}$.
\end{theorem}

\begin{proof}
We first establish truthfulness.
The cost function is public and its breakpoints cannot be controlled by the agents, so we can consider the piece of cake between each pair of consecutive breakpoints separately.
By Theorem~\ref{thm:leximin-truthful}, for each piece, reporting the utility function truthfully yields the highest utility to each agent.
Since the utility for the whole cake is simply the sum of the utilities for different pieces, the mechanism is truthful.

The upper bound of the egalitarian ratio follows from Theorem~\ref{thm:leximin-ratio} since the cost function in the current theorem is more general.
For the lower bound, observe that by Theorem~\ref{thm:leximin-ratio}, for the piece of cake between each pair of consecutive breakpoints, each agent receives a utility of at least a fraction $\frac{\alpha}{n-(n-1)\alpha}$ of her utility for this entire piece of cake.
The desired bound then follows by linearity.
\end{proof}

We remark that since we consider more general cost functions in this section, the egalitarian ratio $\frac{\alpha}{n-(n-1)\alpha}$ is still optimal by Theorem~\ref{thm:impossibility}.
However, unlike the canonical leximin solution, the generalized version is no longer Pareto optimal, since it may be possible to improve the utility of all agents by choosing more than an $\alpha$ fraction in certain parts of the cake and less in other parts.
An interesting question is therefore whether we can obtain Pareto optimality while maintaining truthfulness and the egalitarian ratio.

\end{document}